\documentclass[conference,letterpaper]{IEEEtran}
\newif\iffull
\fulltrue

\usepackage[utf8]{inputenc} 
\usepackage[T1]{fontenc}
\usepackage{url}
\usepackage{ifthen}
\usepackage{cite}
\usepackage[cmex10]{amsmath} 
\usepackage{amsfonts,amssymb,dsfont}
\newtheorem{definition}{Definition}
\newtheorem{theorem}{Theorem}

\newtheorem{lemma}{Lemma}
\newtheorem{remark}{Remark}

\DeclareMathOperator{\supp}{supp}
\usepackage{enumerate}
\usepackage{nicefrac}
\usepackage{tikz}
\usepackage{pgfplots}
\pgfplotsset{compat=1.18}
\usetikzlibrary{positioning,shapes,arrows}
\usepackage{mathtools}

\DeclarePairedDelimiterX{\infdiv}[2]{(}{)}{%
  #1\;\delimsize\|\;#2%
}
\newcommand{\Dkl}{D\infdiv}

\newcommand{\change}[1]{#1}

\usepackage[capitalize,noabbrev]{cleveref}

\IEEEoverridecommandlockouts

\begin{document}
\title{Empirical coordination in the finite blocklength regime: an achievability result\iffull---Extended version\fi} 


\author{%
  \IEEEauthorblockN{Olivier Massicot$^*$, Giulia Cervia$^\dagger$ and Ma\"el Le Treust$^*$}
  \IEEEauthorblockA{$^*$Univ. Rennes, CNRS, Inria, IRISA UMR 6074, F-35000 Rennes, France, Email: \{olivier.massicot, mael.le-treust\}@cnrs.fr \\
  $^\dagger$IMT Nord Europe, Institut Mines-Telecom, Centre for Digital Systems, Lille, Email: giulia.cervia@imt-nord-europe.fr}%
  \thanks{This work was supported by the French National Agency for Research (ANR) via the project n°ANR-22-PEFT-0010 of the France 2030 program PEPR r\'eseaux du Futur.}
}

\maketitle


\begin{abstract}
  Empirical coordination offers a way to understand how agents can coordinate actions under communication constraints. 
  This paper investigates the finite blocklength regime of this problem, where the encoder and decoder aim to produce a sequence of action pairs that is jointly typical with respect to a target distribution. Adopting Shannon’s random coding argument and leveraging the method of types, we analyze the average performance of a random codebook to establish an achievability result. The resulting bound on the optimal rate is presented both in exact form and as an asymptotic expansion, aligning with the prevailing characterizations in the finite blocklength literature. This work extends finite blocklength analysis to the empirical coordination setting, complementing existing results on strong coordination.
\end{abstract}

\section{Introduction}

Empirical coordination in information theory provides a framework to understand how multiple agents can achieve common goals through shared information and coordinated actions under communication constraints. While it traces its origin as far as the common information model of Wyner \cite{1055346}, more recent developments by Cuff et al. \cite{5550277,cuff2011coordination} have delineated the extent of cooperation and coordination in multi-agent networks under communication constraints. 
The question of the coordination of actions was first examined from a game-theoretic perspective by Gossner et al. \cite{gossner2006optimal}.
The empirical coordination framework has since been extended to joint source-channel settings \cite{le2017joint} and to strong coordination over noisy channels with two-sided state information \cite{cervia2020strong}.

Around the same time as Cuff et al.’s work on coordination, Polyanskiy et al. \cite{polyanskiy2010channel} pioneered the study of channel coding in the finite-blocklength regime, soon followed by Kostina and Verdú's work on lossy source coding \cite{kostina2012fixed}. This approach—unlike asymptotic analyses—quantifies fundamental trade-offs in settings where the blocklength is finite, offering insights that are directly relevant to practical systems.

Further developments in the finite-blocklength regime have adapted the approach to wireless communication systems \cite{mary2016finite}, with supporting technical contributions including converse bounds for a range of coding problems \cite{shkel2013converse}, studies of mismatched protection in joint source-channel coding \cite{shkel2014mismatched}, and polar coding solutions for secret key generation \cite{hentila2020polar}. A growing body of recent work has further explored the nonasymptotic regime through risk-aware estimation from compressed data \cite{egan2025risk}, oblivious relaying and variable-length lossy source coding \cite{liu2025nonasymptotic}, rate-distortion-perception trade-offs in distributed computation \cite{gunlu2025low}, and performance limits of secure integrated sensing and communication systems \cite{gunlu2024nonasymptotic}. Complementary to these, new coding strategies have been proposed for consecutive messages with heterogeneous decoding deadlines, through both joint channel coding \cite{nikbakht2022joint} and dirty paper coding \cite{nikbakht2022dirty}.

Adding to the finite blocklength literature in a yet unexplored angle---aside from the work of Cervia et al. \cite{cervia2019fixed} on strong coordination, \change{aiming to produce a joint source and decoded sequences distribution close to an i.i.d. distribution}---, this article addresses the finite blocklength empirical coordination problem, where encoder and decoder aim to produce a sequence of action pairs that is jointly typical with respect to a target distribution. 
\change{Following Shannon's random coding argument \cite{shannon1948mathematical}, also successfully deployed in \cite{kostina2012fixed}, we analyze the average performance of a random codebook using the method of types \cite{csiszar2002method} and deduce an achievability result, thereby circumventing the fact that the conceptual tools developed in \cite{kostina2012fixed} do not readily transpose to the empirical coordination problem.} 
The resulting bound on the optimal rate, stated in exact form in \cref{thm:exactmain} and as an asymptotic expansion in \cref{thm:main}, is reminiscent of the prevailing results in the finite blocklength literature.




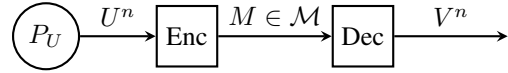
\begin{figure}
\begin{center}
\begin{tikzpicture}[scale=0.8, node distance=1.5cm]
\node (source) at (0,0) [circle,thick,draw,minimum size=0.8cm] {$P_U$};
\node (enc) [right=1cm of source,rectangle,thick,draw,minimum size=0.8cm] {Enc};
\node (dec) [right=of enc,rectangle,thick,draw,minimum size=0.8cm] {Dec};
\node (sink) [right=of dec] {};

\draw [thick,>=stealth,->] (source) -- node[above] {$U^n$} (enc);
\draw [thick,>=stealth,->] (enc) -- node[above] {$M\in\mathcal M$} (dec);
\draw [thick,>=stealth,->] (dec) -- node[above] {$V^n$} (sink);

\end{tikzpicture}
\end{center}
\caption{Rate-constrained empirical coordination: $\mathcal M$ is a fixed finite message set. The objective of the encoder and decoder is to coordinate their actions in order to make the empirical distribution $\hat P_{UV}$ of the pair $(U^n,V^n)$ close to the target distribution $P_{UV}$ as often as possible.}
\label{fig:scheme}
\end{figure}

\subsection{Notation}

$\mathcal U,\mathcal V$ are finite sets, $P_{UV}\in\Delta(\mathcal U\times\mathcal V)$ is a fixed probability distribution, and $U_t$ are i.i.d. random variables with distribution $P_U$ (with $t = 1,2,\dots$). 
When clear from context, we omit the subscript variables and write $P(u)$ instead of $P_U(u)$, likewise for $P_V$, the conditional probability $P_{V|U}$, the joint probability $P_{UV}$, etc. For $u\in\supp P_U$, $P_{V|u}$ denotes the conditional distribution given $u$: $P_{V|u}(v) = P_{V|U}(v|u)$. We let $\mathsf{Var}$ denote the variance of a random variable.

We adopt the notion of typicality defined in \cite[Definition 2.8]{csiszar2011information}: we first let $(\delta_n)$ be any fixed sequence of positive real numbers such that $\delta_n \to_n 0$ yet $n\delta_n^2\to_n\infty$---this is the $\delta$-convention \cite[Convention 2.11]{csiszar2011information}---, then say that a sequence $(u^n,v^n)$ is \emph{jointly typical,} fact denoted $(u^n,v^n) \in \mathcal T^n(UV)$ (without explicit dependence on $\delta_n$), if
\begin{enumerate}[(i)]
  \item for all $(u,v)\in\mathcal U\times\mathcal V$,
  \begin{equation}
    \left|\frac1nN(u,v|u^n,v^n) - P(u,v)\right| \leq \delta_n,
  \end{equation}
  where $N(u,v|u^n,v^n)$ is the number of occurrences of the pair $(u,v)$ in the sequence $(u^n,v^n)$, and
  \item $N(u,v|u^n,v^n)=0$ for all $u,v$ such that $P(u,v)=0$.  
\end{enumerate}

While it is certainly convenient to frame the question in terms of typicality (as we may rely on well-established results), we should also like to stress that the problem can just as well be understood in terms of the empirical distribution of the pair $(U^n,V^n)$, denoted $\hat P$. Explicitly,
\begin{equation}
    \hat P(u,v) = \frac1n N(u,v|U^n,V^n).
\end{equation}
Then $(U^n,V^n)$ is jointly typical if and only if $\supp \hat P \subset \supp P$ and $\|\hat P-P\|_\infty \leq \delta_n$.

\subsection{Formulation of the problem}

An encoder observes the sequence $U^n=(U_1,\dots,U_n)$ and sends a message $M$ to a decoder, which produces an output sequence $V^n$, as depicted in \cref{fig:scheme}. The goal of the encoder and decoder is to coordinate the decoder output with the source sequence, more specifically to make the sequences $(U^n,V^n)$ jointly typical with respect to $P_{UV}$ \emph{as often as possible.} The only constraint is that $n$ is fixed, and the message $M$ must be in a fixed finite set $\mathcal M$.

\begin{definition}
  \label{def:achiev}
  The distribution $P_{UV}$ is $(n,|\mathcal M|,\epsilon)$-achievable with $\epsilon>0$ if there exists a code $(\sigma,\tau)$ producing message $M=\sigma(U^n)$ and output $V^n=\tau(M)$ such that
  \begin{equation}
    \mathbb P[(U^n,V^n) \not\in \mathcal T^n(UV)] \leq \epsilon.
  \end{equation}
  Such a code is called an $(n,|\mathcal M|,\epsilon)$-code. 
  Accordingly, we define the optimal rate as
  \begin{equation}
    \!\mathsf R(n,\epsilon)\!\!
    \triangleq\! \!\min\!\left\{\!\frac1n \!\log \!|\mathcal M|, \!P_{UV}\text{ is $(n,|\mathcal M|,\epsilon)$-achievable}\!\right\}\!.
  \end{equation}
  %
  The dependency on $P_{UV}$ is left implicit henceforth.
\end{definition}

 This article is concerned with achievability results, i.e., upper bounds on $\mathsf R(n,\epsilon)$. The main result is the following.

\begin{theorem}
  \label{thm:main}
  For any $\epsilon\in(0,1)$ fixed,
  \begin{equation}
  \label{eq:mainresult}
    \mathsf R(n,\epsilon)
    \!\leq\! I(U;V) \!+\! \sqrt{\frac{\mathsf{Var}(\mathbb E[\imath(U;V)|U])}{n}}Q^{-1}(\epsilon) \!+\! o\!\left(\frac1{\sqrt n}\right)\!,
  \end{equation}
  where $Q^{-1}$ is the inverse of the Gaussian tail function, the information density $\imath(u;v)$ is defined as
    \begin{equation}
    \imath(u;v)
    \triangleq \log \frac{P_{UV}(u,v)}{P_U(u)P_V(v)}, \quad \forall (u,v) \in \supp P_{UV},
  \end{equation}
  and explicitly,
  \begin{align}
    \label{eq:varEi}
    &\mathsf{Var}(\mathbb E[\imath(U;V)|U]) \nonumber\\
    &\;= \sum_{u\in\supp P_U} P_U(u) \left(\mathbb E[\imath(U;V)|U=u] - I(U;V)\right)^2\!.
  \end{align}
\end{theorem}

\begin{remark}
    The asymptotic bound \eqref{eq:mainresult} is reminiscent of the famous \cite[Theorem 12]{kostina2012fixed}, where the authors tackle the rate-distortion problem and obtain instead the variance of the tilted information ``$\jmath_d(U)$,'' defined in \cite[Definition 6]{kostina2012fixed} linked to optimization program of the rate-distortion function. In our case, no such optimization program is used and the random variable that ``replaces'' the tilted information---$\mathbb E[\imath(U;V)|U]$---stems naturally from the method of types\change{, see \Cref{lem:preBE}}.
\end{remark}

\begin{remark}
    The variance term in \eqref{eq:varEi} is evidently smaller (and thus provide a tighter bound) than $\mathsf{Var}(\imath(U;V))$ which appears when tackling channel coding, e.g., in \cite[Theorem 45]{polyanskiy2010channel}. It can in fact be identified with the difference between the unconditional information variance and the conditional information variance, both coined in \cite[Eq. (239), (242)]{polyanskiy2010channel}:
    \begin{equation}
        \mathsf{Var}(\mathbb E[\imath(U;V)|U])
        = 
        \mathsf{Var}(\imath(U;V))
        - \mathbb E[\mathsf{Var}(\imath(U;V|U))]
    \end{equation}
    In the case $P_{UV}$ corresponds to a capactity-maximizing solution for some channel coding problem, \change{the unconditional and conditional information variances} are equal and thus the variance term in \eqref{eq:mainresult} vanishes, leaving us with the bound $\mathsf R(n,\epsilon)\leq I(U;V)+o(\nicefrac1{\sqrt n})$.
    
    It is noteworthy that different settings still lead to the same form of bounds, albeit with different random variables' variance.
\end{remark}

\begin{remark}
    The asymptotic expansion does not depend on the sequence $(\delta_n)$, so that it holds just as well for other usual notions of typicality such as that of \cite{cover1999elements,el2011network}. Moreover, if the sequence $(\delta_n)$ further satisfies
    \begin{equation}
        \delta_n \geq \frac1{\pi_{P_U}} \sqrt{\frac{\ln n}n}
    \end{equation}
    where $\ln$ is the natural logarithm and $\pi_{P_U}$ is the minimum non-zero probability of $P_U$, then the $o(\nicefrac1{\sqrt n})$ term can be replaced by $O(\nicefrac{\log n}{n})$.
\end{remark}

\subsection{A numerical example}

To illustrate \cref{thm:main}, we have chosen to plot concurrently the asymptotic upper bound \eqref{eq:mainresult3}, denoted $\bar R(n,\epsilon)$, and the exact average performance of a random codebook, denoted $R^\sharp(n,\epsilon)$ on a particular example. See \cref{fig:plot}. These functions are explicitly defined as follows:
\begin{align}
    \bar{R}(n,\epsilon) \triangleq& I(U;V) \!+\! \sqrt{\frac{\mathsf{Var}(\mathbb E[\imath(U;V)|U])}{n}}Q^{-1}(\epsilon),\label{eq:mainresult3}\\
    R^{\sharp}_i(n,\epsilon) \triangleq& \min\Big\{\frac1n \log |\mathcal M|,\label{eq:mainresult2}\\
    &~\mathbb E[(1-\mathbb P[(U^n,\hat V^n)\in\mathcal T_{\delta^i}^n(UV)|U^n])^{|\mathcal M|}] \leq \epsilon\Big\},\nonumber
  \end{align}
where $\hat V^n$ is i.i.d. according to $P_V$ and independent from $U^n$. As such, $R^\sharp(n,\epsilon)$ is defined as the smallest rate $R\geq0$ such that the average probability of error of a random codebook with $|\mathcal M|=\lfloor2^{nR}\rfloor$ generated using $P_U$ is $\epsilon$ or less. This value can be explicitly computed thanks 
\cref{lem:random-codebook}.

The example probability chosen to serve as illustration is $P = \nicefrac13 \delta_{00}+\nicefrac13 \delta_{01}+\nicefrac13 \delta_{10}$, with binary alphabets $\mathcal U = \mathcal V = \{0,1\}$, and with sequences $(\delta_n^i)$ with $\delta_n = \nicefrac{\sqrt{\nicefrac{\ln n}n}}{q_i}$ with $q_1 = 5$, $q_2 = 10$. \Cref{fig:plotdelta} displays the graph of $\delta_n^1$, and \cref{tab:delta} highlights a few sample values.

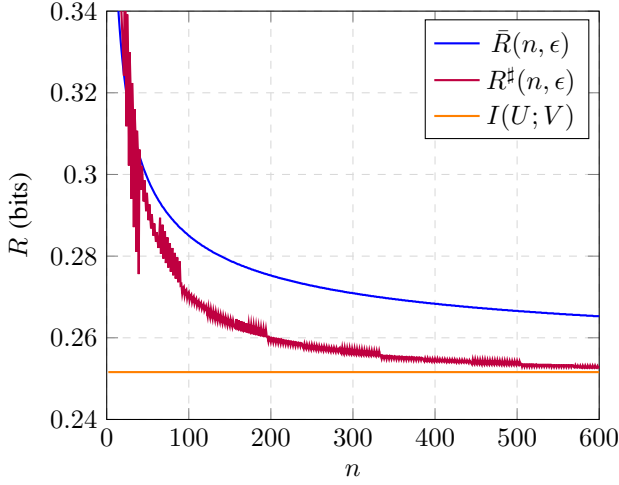
\begin{figure}
\centering
\begin{tikzpicture}
\begin{axis}[
    width=0.9*\columnwidth,
    grid=major, 
    grid style={dashed,gray!30}, 
    xlabel={$n$},
    ylabel={$R$ (bits)},
    xmin=0,
    xmax=400,
    ymin=0.24,
    ymax = 0.34,
]

\addplot[thick, blue]
table [x=n, y=Rapprox, col sep=comma] 
{plot_data.csv};
\addlegendentry{$\bar R(n,\epsilon)$}
\addplot[thick, purple]
table [x=n, y=R1, col sep=comma] 
{plot_data.csv};
\addlegendentry{$R_1^\sharp(n,\epsilon)$};
\addplot[thick, purple!80!green]
table [x=n, y=R2, col sep=comma] 
{plot_data.csv};
\addlegendentry{$R_2^\sharp(n,\epsilon)$};
\addplot[thick, orange]
table [x=n, y=I, col sep=comma] 
{plot_data.csv};
\addlegendentry{$I(U;V)$}
    
\end{axis}
\end{tikzpicture}
\caption{Comparison of the exact average performance of a random code, $R^\sharp_i$, with the asymptotic bound $\bar R$, left-hand side of \eqref{eq:mainresult}.}
\label{fig:plot}
\end{figure}

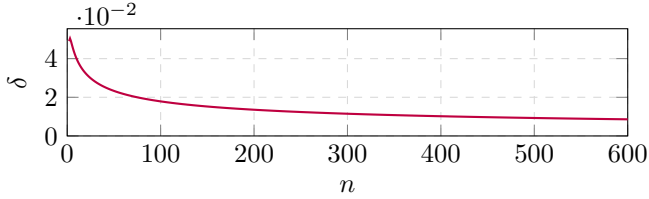
\begin{figure}
\begin{tikzpicture}
\begin{axis}[
    width=\columnwidth,
    grid=major, 
    grid style={dashed,gray!30}, 
    height=3cm,
    xlabel={$n$},
    ylabel={$\delta$},
    xmin=0,
    xmax=400,
    ymin=0,
    xtick distance=100,
    ytick distance=0.1
]

\addplot[thick, purple]
table [x=n, y=d1, col sep=comma] 
{plot_data.csv};
\end{axis}
\end{tikzpicture}
\caption{The sequence $(\delta_n^1)$ chosen: $\delta_n^1=\nicefrac{\sqrt{\nicefrac{\ln n}n}}{5}$.}
\label{fig:plotdelta}
\end{figure}

\begin{table}
\centering
\begin{tabular}{c|cccccc}
  $n$      & 10 & 20 & 40 & 100 & 200 & 400 \\
  \hline
  $\delta$ & 0.096 & 0.077 & 0.061 & 0.043 & 0.033 & 0.024 \\
\end{tabular}
\caption{Example values of $\delta_n^1$.}
\label{tab:delta}
\end{table}

\section{Proof of \cref{thm:main}}

\subsection{Further useful notation}

For a distribution $A\in\Delta(\mathcal X)$ over $\mathcal X$ finite, we denote by $\pi_A$ the minimum non-zero probability of $A$, i.e., $\pi_A = \min_{x\in\supp A} A(x)>0$. We let $\Dkl{\cdot}{\cdot}$ denote the Kullback--Leibler divergence, measured in bits. More generally, base 2 is used throughout this article (with $\log$ in base 2, and $\ln$ in base $e$).

We will occasionally require typicality to be defined for other variables with controlled thresholds. In this case, we will write explicitly the dependence on the threshold and the variables, e.g., $u^n\in\mathcal T^n_\delta(U)$. 

Finally, we will use the notion of $(n,|\mathcal M|,\delta,\epsilon)$-codes and achievability when the typicality threshold is fixed to be $\delta>0$ (instead of $\delta_n$ as in \cref{def:achiev}). Accordingly, we will use the notation $\mathsf R(n,\delta,\epsilon)$ to denote the optimal rate in this case. This notation will be used in the intermediary \cref{thm:exactmain} and in the proof of \cref{thm:main}, noting that $\mathsf R(n,\delta_n,\epsilon) = \mathsf R(n,\epsilon)$.

\subsection{Proof of \cref{thm:main}}

Our article draws many parallels to the article of \cite{kostina2012fixed} that studies the rate-distortion problem in the finite-length regime. The reader familiar with \cite{kostina2012fixed} will easily recognize the three main steps of the achievability proof: \cref{lem:random-codebook} measuring the success of a random codebook, \cref{lem:preBE} lower-bounding the probability of success in terms of the empirical distribution of $\Dkl{P_{V|U}}{P_V} = \mathbb E[\imath(U;V)|U]$, and the Berry--Esseen theorem argument deployed in \cref{thm:exactmain} below.

Among these, the second step is the most original and technical one, and leverages arguments based on the method of types and explicit approximation of types. The technical difficulty of proving \cref{thm:exactmain,thm:main} lies in this step.

\begin{theorem}
  \label{thm:exactmain}
  Let $\mathbb V \triangleq \mathsf{Var}(\mathbb E[\imath(U;V)|U])$ for short and fix \change{a probability error $\epsilon\in(0,1)$ and a typicality threshold $\delta\in(0,1)$.} For all $n\geq \nicefrac4{\pi_U^2 \delta}, \nicefrac2{\pi_{P_{UV}}}, \nicefrac2{\pi_{P_U} \pi_{P_V}}$ such that
  \begin{equation}
    \epsilon > (2|\mathcal U|+1)\exp\!\left(-\nicefrac{n \delta^2 \pi_{P_U}^2}{2}\right) + \nicefrac1{2\sqrt{n\pi_{P_U}}},
  \end{equation}
  it holds that
  \begin{align}
    &\mathsf R(n,\delta,\epsilon) \nonumber\\
    &\leq I(U;V) + 2(|\mathcal U||\mathcal V|+1) \frac1n \log (n+1) + \frac2n 
    \label{eq:mainexactineq}\\
    &+ \sqrt{\frac{\mathbb V}{n}}Q^{-1}\!\left(\epsilon - (2|\mathcal U|+1)\exp(-\nicefrac{n\delta^2\pi_{P_U}^2}2)-\frac1{2\sqrt{n\pi_{P_U}}}\right)\!. \nonumber
  \end{align}
\end{theorem}

The proof of \cref{thm:exactmain} relies on \cref{lem:random-codebook,lem:preBE} in \cref{ssec:randomcodebook,ssec:preBE}, and is stated in \cref{ssec:thmexact}. 

\Cref{thm:exactmain} provides an exact upper bound on $\mathsf R(n,\delta,\epsilon)$ which we use as follows to obtain the asymptotic expansion in \cref{thm:main}. As $n\delta_n^2\to_n\infty$, the conclusion of \cref{thm:exactmain} holds for all $n$ large enough with $\delta=\delta_n$. The term inside $Q^{-1}$ in the bound \eqref{eq:mainexactineq} with $\delta=\delta_n$ is equal to $\epsilon - o_n(1)$. The continuity of $Q$ at $\epsilon>0$ finishes the proof of \cref{thm:main}.



\subsection{Bounding the performance of a random codebook}
\label{ssec:randomcodebook}

In this article, we reconsider the classical random codebook argument of Shannon \cite{shannon1948mathematical} to provide an asymptotic expansion of $\mathsf R(n,\epsilon)$. Accordingly, we let $C=\{\hat V^n_m, ~m\in\mathcal M\}$ be a random codebook (each $\hat V_{t,m}$ is drawn i.i.d. according to $P_V$). 
When there exists a message $m\in\mathcal M$ such that $(U^n,\hat V^n_m)$ are jointly typical, we let the encoder send some such label $M=m$ (say the first such index), and otherwise we let $M$ be random (say the smallest index of $\mathcal M$). The decoder outputs $V^n = \hat V^n_M$.

The stake is thus to upper-bound the probability of ``error,'' i.e., that no message $m$ corresponds to a jointly typical pair of sequences. Let us call $p_C$ this quantity, it is a function of $C$, hence a random variable. If we can show that $\mathbb E[p_C] \leq \epsilon$, then there definitely exists a code $C$ such that $p_C\leq\epsilon$, which thus implements the distribution $P$.

The following lemma, directly adapted from the famous result \cite[Theorem 9]{kostina2012fixed} to suit our needs, quantifies precisely the average performance of the random codebook $C$.

\begin{lemma}
  \label{lem:random-codebook}
  Given two integers $n, |\mathcal M|$ and $\delta\in(0,1)$ all implicitly used in the definition of $C$ and $p_C$,
  \begin{equation}
  \label{eq:lemma1}
    \mathbb E[p_C] = \mathbb E[(1-\mathbb P[(U^n,\hat V^n)\in\mathcal T_\delta^n(UV)|U^n])^{|\mathcal M|}],
  \end{equation}
  where 
  $\hat V_1,\dots,\hat V_n$ are i.i.d. according to $P_V$ and independent from $U^n$.
\end{lemma}


\begin{IEEEproof}[Proof of \cref{lem:random-codebook}]	
	Since the variables $\hat V^n_m$ are i.i.d., the events $(U^n,\hat V^n_m)\not\in\mathcal T^n_\delta(UV)$ \emph{conditional on $U^n$} are mutually independent and all share the same probability, therefore
	\begin{align}
		\mathbb E[p_C]
		&= \mathbb P[\forall m\in \mathcal M, ~(U^n,\hat V^n_m)\not\in\mathcal T_\delta^n(UV)] \\
		&= \mathbb E[\mathbb P[\forall m\in \mathcal M, ~(U^n,\hat V^n_m)\not\in\mathcal T_\delta^n(UV)|U^n]] \\
		&= \mathbb E[\mathbb P[(U^n,\hat V^n_1)\not\in\mathcal T_\delta^n(UV)|U^n]^{|\mathcal M|}] \\
		&= \mathbb E[(1-\mathbb P[(U^n,\hat V^n_1)\in\mathcal T_\delta^n(UV)|U^n])^{|\mathcal M|}] \tag{\ref{eq:lemma1}}.
	\end{align}
\end{IEEEproof}

The second step consists in lower-bounding the probability
\begin{equation}
  \Pi_\delta(u^n) \triangleq \mathbb P[(U^n,\hat V^n)\not\in\mathcal T_\delta^n(UV)|U^n=u^n]
\end{equation}
in terms of the empirical distribution of $\Dkl{P_{V|U_t}}{P_V}$. The following lemma gives such a lower bound, and is based on the method of types. To compare with \cite[Lemma 2]{kostina2012fixed}, the conditional KL-divergence plays the role of the tilted information.

\subsection{Bounding the conditional probability of typical sequences}
\label{ssec:preBE}

\begin{lemma}
\label{lem:preBE}
  For $\delta\in(0,1)$, $n \geq \nicefrac4{\pi_U^2 \delta}, \nicefrac2{\pi_{P_{UV}}}, \nicefrac2{\pi_{P_U} \pi_{P_V}}$ and for $u^n\in\mathcal T_{\nicefrac{\delta \pi_{P_U}}2}^n(U)$, 
  \begin{equation}
    \log \frac1{\Pi_\delta(u^n)}
    \leq 2|\mathcal U||\mathcal V| \log (n+1) + \sum_{t=1}^n \Dkl{P_{V|u_t}}{P_V}.
  \end{equation}
\end{lemma}
  
\begin{IEEEproof}[Proof of \Cref{lem:preBE}]
  In this proof, we use the notion of conditional typicality 
  of \cite[Definition~2.4]{csiszar2011information}: $V^n\in\mathcal T_\delta(V|u^n)$ if 
  \begin{enumerate}[(i)]
  \item $|N(u,v|u^n,v^n) - N(u|u^n) P(v|u)|\leq n\delta$ for all $(u,v)\in\mathcal U\times \mathcal V$, 
  \item $N(u,v|u^n,v^n) = 0$ whenever $P(v|u) = 0$.
  \end{enumerate}
  This has the advantage of explicitly singling out the conditional typicality condition, instead of relying on the notion of joint typicality. One recovers joint typicality in $U,V$ by combining typicality in $U$ and conditional typicality in $V|U$: if $u^n\in\mathcal T^n_{\delta_U}(U)$ and $v^n\in\mathcal T_{\delta_V}(V|u^n)$, then it follows that $(u^n,v^n) \in \mathcal T^n_{(\delta_U+\delta_V)|\mathcal U|}(UV)$, see \cite[Lemma 2.10]{csiszar2011information}. 
%
  We will use this fact with $\delta_U = \delta_V = \nicefrac{\delta\pi_{P_U}}2$ so that $(\delta_U+\delta_V)|\mathcal U|\leq\delta$. 
%
%

  Let now $\delta,n,u^n$ satisfy the premises. Call $T_U\in\Delta(\mathcal U)$ the type of $u^n$ (i.e., $T_U(u) = \nicefrac{N(u|u^n)}{n}$) and $\hat{\mathcal U} = \supp T_U$. 
  %
  For all conditional type $T_{V|U}$ corresponding to a conditionally typical sequence $v^n\in\mathcal T_{\delta_V}(V|u^n)$, we have
  %
  \begin{align}
    \Pi_{\delta}(u^n)
    &\geq \mathbb P[(U^n,\hat V^n) \in \mathbb T^n(T_UT_{V|U})|U^n=u^n] \\
    &\geq \frac1{(n+1)^{|\mathcal U||\mathcal V|}} 2^{-n\Dkl{T_{V|U}}{P_V\,|\,T_U}},
  \end{align}
  where $\mathbb T^n(T_UT_{V|U})$ is the set of sequences $(u^n,v^n)$ of type $T_UT_{V|U}$, and where we use the common notation
  \begin{align}
    \Dkl{T_{V|U}}{P_V\,|\,T_U}
    &= \sum_{u\in\mathcal U} T_U(u) \Dkl{T_{V|u}}{P_V} \\
    &= \frac1n\sum_{t=1}^n \Dkl{T_{V|u_t}}{P_V}.
  \end{align}
  The second inequality, lower bounding the probability of a conditional type, is a standard result of the method of types (see, e.g., \cite[Lemma 2.6]{csiszar2011information}). 

  At this stage, we would like to replace $T_{V|U}$ by $P_{V|U}$. This conditional distribution may however not be a conditional \emph{type.} We resort to the coming approximation lemma (whose proof is given in \iffull\cref{sec:technical}\else\cite[Appendix A]{extended}\fi), where we use the following notation. For a probability distribution $A\in\Delta(\mathcal X)$, we let $\hat{\mathcal T}^n_\delta(A)$ denote the set of types of sequences $x^n\in\mathcal T_\delta(X)$ (with tacitly $P_X=A$), leaving the dependence on $A$ explicit.


  \begin{lemma}
    \label{lem:typeapp}
    For a finite set $\mathcal X$, two probability distributions $A,B \in\Delta(\mathcal X)$, $\delta\in(0,1)$ and any integer $k\geq \nicefrac1{\pi_A},\nicefrac1{\pi_B},\nicefrac1\delta$, 
    \begin{equation}
        \min_{C \in \hat{\mathcal T}_\delta^k(B)} ~\Dkl{C}{A}
        \leq \Dkl{B}{A}  + \frac{|\mathcal X|}k \log k.
    \end{equation}
    Both sides may be positively infinite. 
  \end{lemma}


  Applying \cref{lem:typeapp} with $A=P_V$, $B=P_{V|u}$, $\mathcal X=\mathcal V$, the integer $k=N(u|u^n)\geq n\pi_U - n\delta_U$, and $\delta=\delta_V$ (noting indeed that $N(u|u^n) \geq \nicefrac2{\pi_U \delta}, \nicefrac1{\pi_{P_{V|u}}}, \nicefrac1{\pi_{P_V}}$), we obtain the existence of a conditional type $T_{V|U}$ corresponding to typical sequences $v^n\in\mathcal T_{\delta_V}(V|u^n)$ such that, for all $u\in\hat{\mathcal U}$,
  \begin{align}
    N(u|u^n) \Dkl{T_{V|u}}{P_V} \leq
    &~ N(u|u^n) \Dkl{P_{V|u}}{P_V} \nonumber\\
    &~+ |\mathcal V| \log N(u|u^n),
  \end{align}
  and thus, summing this inequality,
  \begin{equation}
    \Dkl{T_{V|U}}{P_V\,|\,T_U}\!
    \leq\! \sum_{t=1}^n \!\Dkl{P_{V|U_t}}{P_V} \!+\! |\mathcal U||\mathcal V| \!\log n\!.
  \end{equation}

\end{IEEEproof}

\vspace{-0.5cm}

\subsection{Proof of \cref{thm:exactmain}}
\label{ssec:thmexact}

\begin{IEEEproof}[Proof of \cref{thm:exactmain}]
  We let $\delta_U = \delta_V = \nicefrac{\delta\pi_U}{2}$. 
  The first thing to note is that the conditions on $n$ implies that 
  the argument of $Q^{-1}$ belongs to $(0,\epsilon)$. 
  In the remainder, call $R = \frac1n\log|\mathcal M|$ where $|\mathcal M|$ is an integer to fix. As the form of the lower bound on $\Pi_\delta(u^n)$ in \cref{lem:preBE} is similar that used in \cite{kostina2012fixed}, we will leverage similar arguments.

  Thanks to \cref{lem:random-codebook}, and the inequality $(1-x)^\alpha \leq \exp(-\alpha x)$ for $\alpha \geq 1$ and $x\in[0,1]$, 
  \begin{equation}
    \mathbb E[p_C] 
    = \mathbb E[(1-\Pi_\delta(U^n))^{|\mathcal M|}]
    \leq \mathbb E[\exp(-|\mathcal M|\Pi_\delta(U^n))].
  \end{equation}
  Thanks to \cref{lem:preBE} now, for any $\gamma$, 
  \begin{align}
    &\mathbb E[p_C] 
    \leq\mathbb E\!\left[\exp(-|\mathcal M| \Pi_\delta(U^n))\right] \\
    &\leq \exp(-\gamma) + \mathbb P\!\left[\log \frac1{\Pi_\delta(U^n)} > nR-\log\gamma\right] \\
    &\leq \exp(-\gamma) + \mathbb P[U^n \not\in \mathcal T^n_{\delta_U}(U)]\\ 
    &+\!\mathbb P\!\left[\sum_{t=1}^n \Dkl{P_{V|u_t}}{P_V} > nR\!-\!\log\!\gamma\!-\!2|\mathcal U||\mathcal V| \!\log (n+1) \!\right]\!\!. \nonumber
  \end{align}
  Thanks to \cite[Lemma 2.12]{csiszar2011information},
  the middle term is upper-bounded by $2|\mathcal U|\exp(-2n\delta_U^2)$. In the remainder, we use the value $\gamma = 2n\delta_U^2$ and note that $\log\gamma = \log n + \log(2\delta^2) \leq 2\log (n+1)$. As a result,
  \begin{align}
    &\mathbb E[p_C]
    \leq (2|\mathcal U|+1)\exp(-2n\delta_U^2)\label{eq:beforeBE}\\ 
    &+ \!\mathbb P\!\left[\sum_{t=1}^n \Dkl{P_{V|u_t}}{P_V} > nR-2(|\mathcal U||\mathcal V|+1) \log (n+1) \right]\!. \nonumber
  \end{align}

  When $\mathbb V=0$, $\Dkl{P_{V|U_t}}{P_V} = I(U;V)$ for all $t$, thus
  \begin{equation}
    \mathbb E\!\left[\exp(-|\mathcal M| \Pi_\delta(U^n))\right]
    \leq (2|\mathcal U|+1)\exp(-2n\delta_U^2) \leq \epsilon,
  \end{equation}
  as soon as
  \begin{equation}
    R\geq R^* \triangleq I(U;V) + 2(|\mathcal U||\mathcal V|+1) \frac{\log (n+1)}n.
  \end{equation}
  The rate $R^*$ may not correspond to an integer cardinal for the message space but we may invoke \cref{lem:rateapprox} below, thereby concluding the proof for $\mathbb V=0$.

  \begin{lemma}
    \label{lem:rateapprox}
      If $|\mathcal M| \!= \!\left\lceil 2^{nR^*} \right\rceil$, then $R \!\triangleq \!\frac1n \!\log |\mathcal M|
        \leq R^* + \frac2n$.
    \end{lemma}
  
    \begin{IEEEproof}[Proof of \cref{lem:rateapprox}]
      \begin{align}
        R
        &\leq \frac1n \log (2^{nR^*}+1)
        = R^* + \frac1{n\ln 2} \ln (1+2^{-nR^*}) \\
        &\leq R^* + \frac1{n\ln 2} 2^{-nR^*}.
      \end{align}
    \end{IEEEproof}

  When $\mathbb V>0$, the last term of \eqref{eq:beforeBE} can instead be upper-bounded using the Berry--Esseen theorem (see \cite{shevtsova2011absolute}):
  \begin{align}
    &\mathbb P\!\left[nR - 2(|\mathcal U||\mathcal V|+1) \log (n+1) < \sum_{t=1}^n \Dkl{P_{V|U_t}}{P_V}\right] \nonumber \\
    &\quad\leq Q\!\left(\frac{nR - 2(|\mathcal U||\mathcal V|+1) \log (n+1) - nI(U;V)}{\sqrt{n\mathbb V}}\right) \nonumber\\
    &\quad\quad+ \frac{\mathbb E[|\Dkl{P_{V|U}}{P_V}-I(U;V)|^3]}{2\mathbb V^{3/2} \sqrt n}.
  \end{align}
  The following lemma allows to upper-bound the second term. Its proof is given in \iffull\cref{sec:technical}\else\cite[Appendix A]{extended}\fi.

\begin{lemma}
  \label{lem:third-moment-bound}
  Assume $\mathbb V>0$. Then
  \begin{equation}
    \frac{\mathbb E[|\Dkl{P_{V|U}}{P_V}-I(U;V)|^3]}{\mathbb V^{3/2}}
    \leq \frac1{\sqrt{\pi_{P_U}}}.
  \end{equation}
\end{lemma}  

  To manage the first term (and the first term of \eqref{eq:beforeBE}), we would like to set the rate to
  \begin{align}
    R^*
    &\triangleq I(U;V) + 2(|\mathcal U||\mathcal V|+1) \frac1n \log (n+1) \\
    &+ \sqrt{\frac{\mathbb V}{n}}Q^{-1}\!\left(\epsilon - (2|\mathcal U|+1)\exp(-2n\delta_U^2)-\frac1{2\sqrt{n\pi_{P_U}}}\right)\!, \nonumber
  \end{align}
  and obtain a probability of error upper-bounded by $\epsilon$. Unfortunately, this rate may again not correspond to an integer cardinal for the message space, but we may invoke \cref{lem:rateapprox} as before, which leads to the additional term found in the statement of the theorem. This concludes the proof for $\mathbb V>0$. 
\end{IEEEproof}

\section{Exact random codebook performance simulations}

Although it remains computationally intensive to evaluate $\mathsf R(n,\epsilon)$, we can evaluate the average error probability of a random codebook for reasonably small values of $n$ (and in fact $n\delta_n$). We can then compare these exact values to the approximation given by \cref{thm:main} and the upper bound  of \cref{thm:exactmain}.

The exact error probability of a random codebook can be evaluated as follows. First note that $\Pi(u^n)$ only depends on the type of $u^n$, denoted $T_U$ thereafter:
\begin{equation}
  \Pi(u^n)
  = \sum_{\substack{T_{V|U}\in \Delta(\mathcal V|u^n)\\\text{s.t. }T_U T_{V|U} \in \hat{\mathcal T}(P_{UV})}}
  \prod_{u\in\mathcal U} \Lambda(n_u,T_{V|u},P_V),
\end{equation}
where $\hat{\mathcal T}(P_{UV})$ is the set of joint probability distributions $Q$ such that $\|P-Q\|_\infty\leq\delta_n$ and the function $\Lambda$ is defined as follows: $\Lambda(n,T_X,P_X)$ is the probability that the type of an i.i.d. sequence $X^n$ of distribution $P_X$ is equal to $T_X$, and it can be expressed as
\begin{equation}
  \Lambda(n,T_X,P_X)
= \frac{n!}{\prod_{x\in\mathcal X} (n T(x))!} 
  \prod_{x\in\mathcal X} P(x)^{nT(x)},
\end{equation}
see for instance \cite{csiszar2002method}. We will abusively write $\Pi(u^n) = \Pi(T_U)$ to emphasize that $\Pi(u^n)$ only depends on the type $T_U$ of $u^n$.


We can then express the error probability of a random codebook by splitting along the types of $U^n$:
\begin{align}
  &\mathbb E[(1- \Pi(U^n))^{|\mathcal M|}] \nonumber\\
  &\quad= \sum_{T_U\in\Delta_n(P_U)} (1- \Pi(T_U))^{|\mathcal M|} \Lambda(n,T_U,P_U).
\end{align}

With this analysis in hand, we can evaluate the average probability $\mathbb E[(1- \Pi(U^n))^{|\mathcal M|}]$ of error of a random codebook with $n,|\mathcal M|,P$ fixed. We can then find the smallest rate $R$ such that $\mathbb E[(1- \Pi(U^n))^{|\mathcal M|}] \leq \epsilon$ by dichotomy on $|\mathcal M|$. 
The computations of $\Pi(T_U)$ and $\Lambda(n,T_U,P_U)$ only need to be carried out once (for a given $n$), which makes the dichotomy efficient. Moreover, the computations of $\Lambda(n,Q_{V|u},P_V)$ can be stored in a lookup table and reused even as $n$ changes, which further reduces the computational burden. Finally, we note that the number of joint types corresponding to jointly typical sequences grows polynomially with $n \delta_n$, of order $O((n\delta_n)^{|\mathcal U||\mathcal V|-1})$, which makes the above computations tractable for reasonably small values of $n\delta_n$. The code used to produce \cref{fig:plot} is available at \cite{github_script}.



\IEEEtriggeratref{14}






\clearpage
\bibliographystyle{IEEEtran}
\bibliography{bibliofile}

\iffull
\cleardoublepage

\appendices

\section{Proof of the technical lemmas}
\label{sec:technical}



\begin{IEEEproof}[Proof of \Cref{lem:typeapp}]
  In this proof, we let $\mathcal X^* = \supp B$. If $B\not\ll A$, there is nothing to prove as the right-hand side is positively infinite. We assume thus that $\mathcal X^* \subset \supp A$. For any $C\in\Delta(\mathcal X)$ with support equal to $\mathcal X^*$, by convexity of $\Dkl{\cdot}{A}$,
  \begin{equation}
    \Dkl{C}{A}
    \leq \Dkl{B}{A} + \sum_{x\in\mathcal X^*} (C(x)-B(x)) \log\frac{C(x)}{A(x)}.
  \end{equation}
  Consider taking $C^*$ with $C^*(x)=0$ for $x\in\mathcal X\setminus\mathcal X^*$ and $C^*(x) = \frac{\lfloor kB(x) \rfloor + \epsilon_x}{k}$ for $x\in\mathcal X^*$ where $\epsilon_x \in \{0,1\}$ is to be adjusted. Note that $k\pi_B\geq1$ implies that $\lfloor kB(x) \rfloor\geq1$ for all $x\in\mathcal X$, hence $\supp C^* = \supp B$ and $\pi_{C^*}\geq\nicefrac1k$. For $C^*$ to be a type with $k$ samples, it suffices that
  \begin{equation}
    \sum_{x\in\mathcal X^*} \frac{\lfloor kB(x) \rfloor + \epsilon_x}{k} = 1,
  \end{equation}
  i.e., that
  \begin{equation}
    \sum_{x\in\mathcal X^*} \epsilon_x 
    = k - \sum_{x\in\mathcal X^*} \lfloor kB(x) \rfloor.
  \end{equation}
  The right-hand side is an integer between $0$ and $k$, hence there is a choice of $\epsilon$ that makes $C^*$ a type. We now fix $\epsilon$ as such. Surely, for all $x\in\mathcal X$,
  \begin{equation}
    |C^*(x)-B(x)| \leq \frac1k \leq \delta.
  \end{equation}
  As a result, $C^* \in \mathcal B_\delta^k(B)$, and
  \begin{align}
    &\min_{C \in \hat {\mathcal T}_\delta^k(B)} \Dkl{C}{A} \\
    &\quad\leq \Dkl{C^*}{A} \\
    &\quad\leq \Dkl{B}{A} + \sum_{x\in\mathcal X^*} (C^*(x)-B(x)) \log\frac{C^*(x)}{A(x)} \\
    &\quad\leq \Dkl{B}{A} + \sum_{x\in\mathcal X^*} \frac1k \log \max\!\left(k,\frac1{\pi_A}\right). 
  \end{align}
\end{IEEEproof}

\begin{IEEEproof}[Proof of \Cref{lem:third-moment-bound}]
  Let $Z = \Dkl{P_{V|U}}{P_V}-I(U;V)$, so that $\mathbb E[Z]=0$ and $\mathbb E[Z^2]=\mathbb V$. Since $\mathcal U$ is finite, the maximum
  \begin{equation}
    \mathbb D \triangleq \max_{u\in\supp P_U} |\Dkl{P_{V|U}}{P_V}-I(U;V)|
  \end{equation}
  is finite and is attained at some $u^\star\in\supp P_U$. Therefore,
  \begin{align}
    \mathbb V &= \sum_{u\in\supp P_U} P_U(u)\,|\Dkl{P_{V|U}}{P_V}-I(U;V)|^2 \\
    &\geq P_U(u^\star) \mathbb D^2 
    \geq \pi_{P_U} \mathbb D^2.
  \end{align}
  Hence
  \begin{align}
    \mathbb E[|Z|^3]
    &= \sum_{u\in\supp P_U} P_U(u)\,|\Dkl{P_{V|U}}{P_V}-I(U;V)|^3 \\
    &\leq \mathbb D \sum_{u\in\supp P_U} P_U(u)\,|\Dkl{P_{V|U}}{P_V}-I(U;V)|^2 \\
    &= \mathbb D\,\mathbb V
    \leq \frac{\mathbb V^{3/2}}{\sqrt{\pi_{P_U}}}.
  \end{align}
  Dividing by $\mathbb V^{3/2}$ proves the claim.
\end{IEEEproof}

\fi

\end{document}